\documentclass[a4paper]{article}

\usepackage[english]{babel}
\usepackage[utf8x]{inputenc}
\usepackage[T1]{fontenc}

\usepackage[a4paper,top=3cm,bottom=2cm,left=3cm,right=3cm,marginparwidth=1.75cm]{geometry}

\usepackage{amsmath}
\usepackage{graphicx}
\usepackage{caption}
\usepackage{subcaption}
\usepackage[colorinlistoftodos]{todonotes}
\usepackage[colorlinks=true, allcolors=blue]{hyperref}
\usepackage{amssymb}
\usepackage{amsthm}
  \newtheorem{lemma}{Lemma}
  \newtheorem{theorem}{Theorem}
  \newtheorem{corollary}{Corollary}
    
  \newtheorem{observation}{Observation}
  
  \newtheorem{remark}{Remark}
  \theoremstyle{definition}
  \newtheorem{definition}{Definition}
\usepackage{enumitem}
\usepackage{verbatim}
\usepackage{cases}

\newcommand{\cC}{\mathcal{C}}
\newcommand{\cP}{\mathcal{P}}

\newcommand{\cB}{\mathcal{B}}
\newcommand{\cS}{\mathcal{S}}
\newcommand{\cI}{\mathcal{I}}

\newcommand{\precs}{\prec_s}
\newcommand{\width}[1]{\textup{width}(#1)}
\newcommand{\pathwidth}[1]{\textup{pw}(#1)}

\newcommand{\cpw}[1]{\textup{cpw}(#1)}
\newcommand{\st}{\hspace{0.1cm}\bigl|\bigr.\hspace{0.1cm}}

\newcommand{\card}[1]{\left|#1\right|}
\newcommand{\branches}[1]{\cB(#1)}
\newcommand{\branchesin}[2]{\cB^{#1}(#2)}
\newcommand{\bag}[1]{\textup{bag}(#1)}
\newcommand{\cover}[1]{\textup{cover}(#1)}
\newcommand{\border}[1]{\textup{border}(#1)}
\newcommand{\botStart}[2]{t_1(#1,#2)}
\newcommand{\botEnd}[2]{t_2(#1,#2)}

\newcommand{\cMin}[1][]{%
\ifthenelse{\equal{#1}{}}{c_{\min}}{c_{\min}(#1)}%
}
\newcommand{\graphStart}[2]{\alpha(#1,#2)}
\newcommand{\graphEnd}[2]{\beta(#1,#2)}
\newcommand{\startv}{s^*}

\newcommand{\graphStartS}[1]{\alpha(#1)}
\newcommand{\graphEndS}[1]{\beta(#1)}
\newcommand{\botStartS}[1]{t_1(#1)}
\newcommand{\botEndS}[1]{t_2(#1)}

\newcommand{\preB}{\textup{pre}}
\newcommand{\inB}{\textup{in}}
\newcommand{\postB}{\textup{post}}

\begin{document}

\renewcommand{\thefootnote}{\fnsymbol{footnote}}
\title{Finding small-width connected path decompositions in polynomial time}
\author{Dariusz Dereniowski\footnotemark[1] \and Dorota Osula\footnotemark[1] \and Paweł Rzążewski\footnotemark[2]}
\footnotetext[1]{Faculty of Electronics, Telecommunications and Informatics, Gda{\'n}sk University of Technology, Gda{\'n}sk, Poland}
\footnotetext[2]{Faculty of Mathematics and Information Science, Warsaw University of Technology, Warsaw, Poland}
\date{}

\maketitle

\begin{abstract}
A connected path decomposition of a simple graph $G$ is a path decomposition $(X_1,\ldots,X_l)$ such that the subgraph of $G$ induced by $X_1\cup\cdots\cup X_i$ is connected for each $i\in\{1,\ldots,l\}$.
The connected pathwidth of $G$ is then the minimum width over all connected path decompositions of $G$.
We prove that for each fixed $k$, the connected pathwidth of any input graph can be computed in polynomial-time.
This answers an open question raised by Fedor V. Fomin during the GRASTA 2017 workshop, since connected pathwidth is equivalent to the connected (monotone) node search game.
\end{abstract}


\section{Introduction} \label{sec:intro}

Since the famous `graph minor' project by Robertson and Seymour that started with \cite{RobertsonS83}, the notions of treewidth and pathwidth received growing interest and a vast amount of results has been obtained.
The pathwidth, informally speaking, allows us to say how closely an arbitrary graph resembles a path.
This concept proved to be useful in designing algorithms for various graph problems, especially in the case when the pathwidth of an input graph is small (e.g. fixed), in which case quite often a variant of well-known dynamic programming approach that progresses along a path decomposition of input graph turns out to be successful.

Several modifications to pathwidth have been proposed and in this work we are interested in the connected variant in which one requires that a path decomposition $(X_1,\ldots,X_l)$ of a graph $G$ satisfies: the vertices $X_1\cup\cdots\cup X_i$ induce a connected subgraph in $G$ for each $i\in\{1,\ldots,l\}$.
This version of the classical pathwidth problem is motivated by several pursuit-evasion games, including, but not limited to, \emph{edge search}, \emph{node search} or \emph{mixed search}~\cite{searchingPebbling,ParsonsPursuitEvasion,takahashi1995mixed}.
More precisely, computing the minimum number of searchers needed to clean a given graph $G$ in the node search game (i.e., computing the \emph{node search number} of $G$) is equivalent to computing the connected pathwidth of $G$. Moreover, a connected path decomposition can be easily translated into the corresponding node search strategy that cleans $G$ and vice versa.
For further references that provide more details on correspondences between connected path decomposition and different variants of the search games see e.g.~\cite{BFFFNST12,Dereniowski12SIDMA,FN06,FT08}.

\subsection{Related work} \label{sec:related}

A lot of research has been done in the direction of obtaining FPT algorithms for pathwidth, parametrized by the pathwidth $k$.
One of the first polynomial-time algorithm was presented in $1983$ by Ellis \emph{et al.} and had running time $O(n^{2k^2+4k+2})$~\cite{EllisST83}. Later these results were improved (e.g., \cite{Bodlaender96,BodlaenderK96}) leading to the currently fastest FPT algorithm, working in time $2^{O(k^2)}n$~\cite{Furer16}. In these algorithms, in order to produce an optimal solution, an approximate path- or tree-decomposition is pre-computed.
It is thus of interest to have good approximations for these problems. Numerous works have been published in this direction \cite{Amir10,Lagergren96,Reed92}, leading to the currently fastest algorithm for constant-factor approximation for treewidth, working in time $2^{O(k)}n$~\cite{BodlaenderDDFLP16}, that is, single exponential in the treewidth $k$. The best known approximation ratio of a polynomial time approximation algorithm for pathwidth is $O(\sqrt{\log (opt)}\log n)$~\cite{FeigeHL08}.

There exist exact algorithms for computing pathwidth, whose running times are exponential in the order of the input graph.
Pathwidth can be computed in $O^*(2^n)$-time (in $O^*(2^n)$ space) or in $O^*(4^n)$-time with the use of polynomial space, using a simple algorithm from~\cite{BodlaenderFKKT12}.
There is also a faster algorithm with running time $O^*(1.9657^n)$~\cite{SuchanV09}, which has been further improved very recently to $O^*(1.89^n)$~\cite{KitsunaiKKTT12}.
See \cite{BiedlBNNPR13,Coudert16,CoudertMN16} for some experimental approaches to pathwidth computation.

For pathwidth, it is known due to~\cite{RobertsonS04} that the set of minimal forbidden minors (i.e., the obstruction set) is finite for each fixed $k$.
However, a significant difference between pathwidth and connected pathwidth is that the latter one is not closed under taking minors and hence it is not known if the set of minimal forbidden minors for connected pathwidth is finite~\cite{BestGTZ16}.

We also point out that a number of results have been obtained for connected pathwidth or the closely related connected graph (edge) search, including algorithmic and computational ones \cite{BFFFNST12,Dereniowski11,Dereniowski12SIDMA,connected_outerplanar,FN06,Nisse08}, monotonicity \cite{FN08,YDA09}, structural properties \cite{connected_and_internal} or distributed algorithms \cite{BorowieckiDK15,FHL05,INS09}.

\subsection{Motivation} \label{sec:motivation}

The connectivity constraint for pathwidth is natural and useful in graph searching games \cite{BFFFNST12,FHL05,FN06}.
The connectivity is in some cases implied by potential applications (e.g., security constraints may enforce the clean, or safe, area to be connected) or it is a necessity, like in distributed or online versions of the problem~\cite{BlinFNV08,BorowieckiDK15,INS09,NS09}.

Our second motivation comes from connections between pathwidth and connected pathwidth.
More specifically, \cite{Dereniowski12SIDMA} implies that for any graph $G$, these parameters differ multiplicatively only by a small constant.
This implies that an approximation algorithm for connected pathwidth immediately provides an approximation algorithm for pathwidth with asymptotically the same approximation ratio.
This may potentially lead to obtaining better approximations for pathwidth since, informally speaking, the algorithmic search space for connected pathwidth is for some graphs much smaller than that for pathwidth.
On the other hand, we do not know any algorithm computing the connected pathwidth in time $O^*((2-\epsilon)^n)$, for any $\epsilon > 0$. Thus, despite this smaller algorithmic search space, it is not clear how these two problems algorithmically differ in the context of designing exact algorithms.

During the GRASTA 2017 workshop, Fedor V. Fomin \cite{fominQuestion} raised an open question, whether we can verify in polynomial time, if the connected pathwidth of the given graph is at most $k$, for a fixed constant $k$. In this paper we answer this question in the affirmative.

\subsection{Outline} \label{sec:our-work}
In the next section we recall a definition of connected pathwidth and related terms used in this work.
Section~\ref{sec:algorithm} provides a polynomial-time algorithm for determining whether the connected pathwidth of an arbitrary input graph $G$ is at most $k$, where $k$ is a fixed integer. The algorithm is inspired by the algorithms for computing mininum-length path decompositions by Dereniowski, Kubiak, and Zwols~\cite{DERENIOWSKI20151715}, we also use the notation from this paper.
Then, Section~\ref{sec:analysis} contains the analysis of the algorithm (its correctness and running time).
We finish with some open problems in Section~\ref{sec:open-problems}.

\section{Definitions} \label{sec:definitions}
For a simple graph $G=(V(G),E(G))$ and a set $Y\subseteq V(G)$, the subgraph with vertex set $Y$ and edge set $\{\{u,v\}\in E(G)\st u,v\in Y\}$ is denoted by $G[Y]$ and is called the subgraph \emph{induced} by $Y$.
For $Y\subseteq V(G)$, we write $N_G(Y)$ to denote the \emph{neighborhood} of $Y$ in $G$, defined as $N_G(Y)=\{v\in V(G)\setminus Y\st \exists u\in Y \textup{ s.t. }\{u,v\}\in E(G)\}$. 

\begin{definition} \label{def:path_dec}
A \emph{path decomposition} of a simple graph $G=(V(G),E(G))$ is a sequence $\cP=(X_1,\ldots,X_l)$, where $X_i\subseteq V(G)$ for each $i\in\{1,\ldots,l\}$, and
\begin{enumerate}[label={\normalfont(\roman*)},align=left,leftmargin=*]
 \item\label{it:path1} $\bigcup_{i=1}^{l} X_i=V(G)$, 
 \item\label{it:path2} for each $\{u,v\}\in E(G)$ there exists $i\in\{1,\ldots,l\}$ such that $u,v\in X_i$,
 \item\label{it:path3} for each $i,j,k$ with $1\leq i\leq j\leq k\leq l$ it holds that $X_i\cap X_k\subseteq X_j$.
\end{enumerate}
The \emph{width} of a path decomposition $\cP$ is $\width{\cP}=\max_{i\in\{1,\ldots,l\}}|X_i|-1$.
The \emph{pathwidth} of $G$, denoted by $\pathwidth{G}$, is the minimum width over all path decompositions of $G$.
\end{definition}

We say that a path decomposition $\cP=(X_1,\ldots,X_l)$ is \emph{connected} if the subgraph $G[X_1\cup\cdots\cup X_i]$ is connected for each $i\in\{1,\ldots,l\}$. A \emph{connected pathwidth} of a graph $G$, denoted by $\cpw{G}$, is the minimum width taken over all connected path decompositions of $G$.

Finally, a \emph{connected partial path decomposition} of a graph $G$ is a connected path decomposition $(X_1,\ldots,X_i)$ of some subgraph $H$ of $G$, where $N_G(V(G)\setminus V(H))\subseteq X_i$.
In other words, in the latter condition we require that each vertex of $H$ that has a neighbor outside $H$ belongs to the last bag $X_i$.
Intuitively, a connected partial path decomposition of $G$ can be potentially a prefix of some connected path decomposition of $G$.
Also note that if $(X_1,\ldots,X_l)$ is a connected path decomposition, then for each $i\in\{1,\ldots,l\}$, $(X_1,\ldots,X_i)$ is a connected partial path decomposition of $G$, where $H=G[X_1\cup\cdots\cup X_i]$.

In our analysis we will use an intermediate notion between arbitrary and connected path decompositions.
For a path decomposition $\cP=(X_1,\ldots,X_l)$ of $G$ and $\cI\subseteq V(G)$, we say that it is \emph{$\cI$-connected} if for each $i\in\{1,\ldots,l\}$ each connected subgraph of $G[X_1\cup\cdots\cup X_i]$ contains a vertex from $\cI$.
In other words, if one takes a prefix of $\cP$, then the subgraph of $G$ induced by the vertices of this prefix may have several connected components and each of them must have a vertex in $\cI$.
Note that if $\cP$ is connected then there is only one such connected component.

\section{The algorithm} \label{sec:algorithm}
Note that we may without loss of generality assume that the first bag in a connected path decomposition to be computed has only one vertex.
In the remaining part of the paper we denote this vertex by $\startv$. For the time being we may assume that $\startv$ is known to the algorithm because checking all possible values of $\startv\in V(G)$ increases the complexity only by a factor of $n$.
We will present an algorithm  deciding whether, for an input graph $G$ and $\cI=\{\startv\}$, there exists an $\cI$-connected path decomposition of width $k-1$.
Our algorithm is recursive so for subsequent recursive calls all three parameters are different.
We will consider width $k-1$ so that the sets we have to handle will have size at most $k$.
This modification in which we consider the `starting bag' of a decomposition is dictated by the fact that our algorithm is recursive and on subsequent levels of the recursion the content of the first bag cannot be arbitrary in connected path decompositions (as opposed to just path decompositions).

In the following by $G$ we denote the input graph with $n$ vertices, $\cI$ and $k$ are also fixed.
In the rest of the work, $k$ refers to the maximum bag size of the connected path decomposition to be computed (thus, the algorithm checks whether $\cpw{G}\leq k-1$ for the input graph $G$).

For any $S\subseteq V(G)$, we say that a subgraph $H$ of $G$ is an \emph{$S$-branch} if $H$ is a connected component of $G-S$ and $N_G(V(H))=S$. For any $S\subseteq V(G)$, define $\branches{S}$ to be the set of all $S$-branches. 
A set $S$ is called a \emph{bottleneck} if the number of $S$-branches is at least $2k+1$, as it guarantees us the existence of at least one special branch called an \emph{in-branch}, which will be defined formally in the next Section.
Observe that each connected component of $G-X$, for any $X\subseteq V(G)$, is an $S$-branch for exactly one non-empty subset $S$ of $X$.

Let us mention that $S$-branches are also known as \emph{full components associated with $S$} (see e.g. Bouchitt{\'e} and Todinca~\cite{DBLP:journals/siamcomp/BouchitteT01}).

\subsection{States}

By a \emph{potential state} we mean a triple $(X,\{B_S\}_{S \subseteq X}, \{f^B_S \}_{S \subseteq X})$, consisting of:
\begin{itemize}
\item a non-empty set $X \subseteq V(G)$ with $\card{X} \leq k$,
\item a subset $B_S \subseteq \branches{S}$ of cardinality at most $2k$, chosen for every non-empty $S \subseteq X$,
\item a function $f^B_S \colon \branches{S} \to \{0,1\}$, chosen for every non-empty $S \subseteq X$. We additionally require that if $H,H' \in \branches{S} \setminus B_S$, then $f^B_S(H)=f^B_S(H')$.
\end{itemize}
The exact meaning of $B_S$ will be explained later on, but let us present some intuition. In Lemma~\ref{lem:small} we show that for every path decomposition $\cP$, the vertices of all but at most $2k$ $S$-branches appear in bags of $\cP$ in a certain, well-structured way. The set $B_S$ and the function $f^B_S$ will used to describe the structure of the remaining, badly behaving $S$-branches.

Observe that the set $X$ may be chosen in at most $n^k$ ways and the number of choices of $S$ is at most $2^k$. For every $S$, the number of $S$-branches is at most $n$, so $B_S$ can be chosen in at most $n^{2k}$ ways. The function $f^B_S$ can be chosen in at most $2^{|B_S|} \cdot 2 \leq 2^{2k+1}$ ways. Therefore the number of potential states is at most $n^k  \cdot 2^k \cdot n^{2k} \cdot 2^{2k+1}= O(n^{3k})$, i.e., polynomial in $n$, where the asymptotic notation hides the factor that depends on $k$.

With a potential state $s=(X,\{B_S\}_{S\subseteq X},\{f^B_S\}_{S \subseteq X})$ we associate the following notions.
By $\bag{s}$ we denote the set $X$.
By $\cover{s}$ we denote the set of vertices
\[X \cup \bigcup_{S \subseteq X, S \neq \emptyset} \left( \bigcup_{H \in \branches{S} \colon f_S^B(H)=1} V(H) \right).\]
By $G_s$ we denote the subgraph of $G$ induced by $\cover{s}$.
We say that two states $w,s$ are \emph{indistinguishable} if $\cover{w}=\cover{s}$ and $\bag{w} = \bag{s}$. Otherwise states are \emph{distinguishable}.
We note that for such two distinguishable states it may hold that $\cover{w}=\cover{s}$ but $\bag{w}\neq\bag{s}$.

Let $v$ be a vertex from $\cover{s}$ which has a neighbor $u \notin \cover{s}$.
We argue that $v\in\bag{s}$.
Otherwise, if $v \notin \bag{s}$, then both $v$ and $u$ belong to the same $S$-branch for some $S \subseteq \bag{s}$. Thus, they are either both in $\cover{s}$, or outside it. From this it follows that every vertex $v \in \cover{s}$, which has a neighbor $u \notin \cover{s}$, must belong to $\bag{s}$. Let us denote the set of such vertices $v$ by $\border{s}$.
We have proved:
\begin{observation} \label{obs:cover}
For each potential state $s$ it holds that $\border{s}\subseteq\bag{s}$.
\end{observation}

We say that a potential state $s$ is a \emph{state} if each connected component of $G_s$ contains a vertex in $\cI$.

We introduce a boolean table $Tab$, indexed by all states.
For a state $s$, the value of $Tab[s]$ will be set to $true$ by our algorithm if and only if there exists some $\cI$-connected path decomposition $\cP=(X_1,X_2,\ldots,X_l)$ of $G_s$, such that $\width{\cP}$ is at most $k-1$, $\cP$ starts with $\cI$ and $X_l = \bag{s}$. We will use a dynamic programming to fill out the table $Tab$. Then, we will conclude $\cpw{G} \leq k-1$ (with the additional constraint concerning $\cI$) if and only if $Tab[s] = true$ for some state $s$ with $\cover{s} = V(G)$. 

Observe that such a final state exists, since for $s = (X, \{B_S\}_S, \{f^B_s\}_S)$, we have $\cover{s}=V(G)$ if and only if $f^B_S(H)=1$ for every $S$ and $H$.
However, the astute reader may notice that in our representation we might have not included some $\cI$-connected partial path decompositions and it could be possible that we do not find a solution, even though it exists. We will show that if $\cpw{G} \leq k-1$, then there exists a special type of a $\cI$-connected path decomposition of width at most $k-1$, called a \emph{structured path decomposition} (defined later), which can be found using our algorithm because, as we will argue, our table $Tab$ does not `omit' any structured path decompositions.

\subsection{Extension rules}

Let us introduce a total ordering $\precs$ on the set of states. We say that $w \precs s$ if $|\cover{w}| < |\cover{s}|$, or $|\cover{w}|=|\cover{s}|$ and $|\bag{w}| > |\bag{s}|$. If $|\cover{w}| = |\cover{s}|$ and $|\bag{w}| = |\bag{s}|$, then we resolve such a tie arbitrarily.

In our dynamic programming, we process states according to the ordering $\precs$ and fill the table $Tab$ using two extension rules: \emph{step extension} and \emph{jump extension}.
 
\begin{description}
\item[step extension] for distinguishable states $w, s$ with $s \prec_s w$ and $|\cover{s}|>1$: if $Tab[w] = true$ and
\begin{enumerate}[label={\normalfont(S\arabic*)},align=left,leftmargin=*]
\item\label{step:1} each connected component of $G[\bag{s}]$ contains a vertex from $\bag{w}$,
\item\label{step:2} $\border{w}\subseteq\bag{s}$,
\item\label{step:3} $\cover{s} = \cover{w} \cup \bag{s}$,
\item\label{step:4} $\bag{s} \cap \cover{w} \subseteq \bag{w}$,
\end{enumerate}
then set $Tab[s]$ to $true$,
\item[jump extension] for distinguishable states $w = (X,\{B_S\}_{S\subseteq X},\{f^B_S\}_S)$ and $s = (X,\{B_S\}_{S\subseteq X},\{g^B_S\}_S)$ with $s \prec_s w$ and $|\cover{s}|>1$: if $Tab[w] = true$ and there exists a bottleneck set $S' \subseteq X$, such that
\begin{enumerate}[label={\normalfont(J\arabic*)},align=left,leftmargin=*]
\item\label{jump:green} $f^B_{S'}(H) = 0$ and $g^B_{S'}(H) = 1$ for every $H \in \branches{S'} \setminus B_{S'}$,
\item\label{jump:otherS} $f^B_{S'}(H) = g^B_{S'}(H)$ for every $H \in B_{S'}$,
\item\label{jump:other} $f^B_{S}(H) = g^B_{S}(H)$ for every non-empty $S \neq S'$, $S\subseteq X$, and $H \in \branches{S}$,
\item\label{jump:rec} for each $H \in \branches{S'} \setminus B_{S'}$ there exists a $N_G(S')\cap V(H)$-connected path decomposition $\cP_H$ of $H$ of width at most $k-|X|-1$,
\end{enumerate}
then set $Tab[s]$ to $true$.
\end{description}

Let us present some intuitions behind these extension rules.
In step extension, if $w$ corresponds to some $\cI$-connected partial path decomposition $\cP = (X_1,X_2,\ldots,X_l=\bag{w})$, then $s$ corresponds to an $\cI$-connected partial path decomposition $\cP' = (X_1,X_2,\ldots,X_l,X_{l+1}=\bag{s})$ (we extend $\cP$ by adding a single bag, namely $\bag{s}$).
Also note, that each new vertex in $X_{l+1}$, i.e. one that is not in $X_l$, has (due to~\ref{step:1}) a neighbor in $X_l$, as required in connected path decompositions.

In jump extension, if the state $w$ corresponds to some $\cI$-connected partial path decomposition $\cP = (X_1,X_2,\ldots,X_l)$, then the state $s$ corresponds to a $\cI$-connected partial path decomposition $\cP' = (X_1,X_2,\ldots,X_l,X_{l+1},X_{l+2},\ldots,X_{l + l'})$, where
\begin{itemize}
\item $\bigcap_{i=l}^{l+l'}X_i = X_l = X_{l+l'}$,
\item $\cP'' := (X_{l+1} \setminus X_l, X_{l+2} \setminus X_l,\ldots, X_{l+l'} \setminus X_l)$  is a (not necessarily connected) path decomposition of the graph induced by $S'$-branches, for some $S' \subseteq X_l$.
These are the $S'$-branches $H$ in~\ref{jump:green} and $\cP''$ is obtained by `concatenating' the path decompositions from~\ref{jump:rec}.
\end{itemize}
Note that although a path decomposition $\cP_H$ in~\ref{jump:rec} may not be connected, we ensure (by definition of $\cI$-connectivity) that each connected component of the subgraph induced by each prefix of $\cP_H$ has a vertex from $N_G(S')\cap V(H)$ and hence it contains a neighbor of $S'$, ensuring required connectivity of the resulting path decomposition.

\subsection{Summing up}
 Let us recall how the algorithm works. We introduce a boolean table $Tab$, indexed by all states. We initialize $Tab$ by setting $Tab[s]=true$ for every state $s$, such that $\cover{s} = \bag{s} = \{ v \}$, for some $v \in \cI$, while for the remaining states $s$ we initialize $Tab[s]$ to be $false$.

Then we process states in the ordering $\precs$, checking whether a step extension or a jump extension can be applied to set $Tab[s]= true$.
We terminate when we find a state corresponding to a feasible solution (i.e., we set $Tab[s]=true$ for some state $s$ with $\cover{s}=V(G)$), or when we process all states, in the latter case we report that a solution does not exist.
 
 We note that the algorithm will be subsequently switching some entries of $Tab$ from $false$ to $true$, and hence until the completion of the algorithm it is understood that the value $false$ of a particular entry of $Tab$ does not provide any information as to whether a corresponding path decomposition exists.

\section{The analysis} \label{sec:analysis}

Let us start by introducing some more definitions and additional notation.
Let $\cP=(X_1,X_2,\ldots,X_l)$ be a path decomposition of $G$.
We say that a connected subgraph $H$ of $G$ is \emph{contained in} an interval $[i,j]$ of $\cP$ for some $1\leq i\leq j\leq l$, if $V(H)\cap X_t\neq\emptyset$ if and only if $t\in\{i,\ldots,j\}$. Note that this definition is valid since it follows that for a connected subgraph $H$, the subset of indices $t$ such that $V(H)\cap X_t\neq\emptyset$ is indeed an interval. If $H$ is contained in $[i,j]$, then we denote these endpoints of the interval as $i=\graphStart{H}{\cP}$ and $j=\graphEnd{H}{\cP}$.
If $\cP$ is clear from the context, we will often write shortly $\graphStartS{H} := \graphStart{H}{\cP}$ and $\graphEndS{H}:=\graphEnd{H}{\cP}$.

For a set  $S$ and an $\cI$-connected path decomposition $\cP=(X_1,X_2,\ldots,X_l)$, we say that an $S$-branch $H$ is an \emph{in-branch} if $S \subseteq X_{\graphStartS{H}}$ and $S \subseteq X_{\graphEndS{H}}$. The lemma below gives us a lower bound for number of in-branches of $S$.

\begin{lemma}\label{lem:small}
For every set $S$ and path decomposition $\cP=(X_1,X_2,\ldots,X_l)$, at most $2k$ $S$-branches are not in-branches.
\end{lemma}
\begin{proof} 
Consider an $S$-branch $H$, which is not an in-branch. This means that $S \not\subseteq X_{\graphStartS{H}}$ or $S \not\subseteq X_{\graphEndS{H}}$.

First, consider $H$, such that $S \not\subseteq X_{\graphStartS{H}}$. Let $t$ be the minimum index such that $S \subseteq X_1 \cup X_2 \cup \cdots \cup X_t$. Let $v \in S$ be a vertex in $X_t \setminus (X_1 \cup X_2 \cup \cdots \cup X_{t-1})$, it exists by the definition of $t$.
Recall that $v$ is a neighbor of some vertex $w$ of $H$, so, since $\cP$ is a path decomposition, there must be a bag $X_i$ containing both $v$ and $w$. Since $X_t$ is the first bag, where $v$ appears, we observe that $i \geq t$ and thus $\graphEndS{H} \geq t$.

Suppose now that $\graphStartS{H} > t$. Note that since $S \not\subseteq X_{\graphStartS{H}}$, there is some $u \in S \setminus X_{\graphStartS{H}}$.
Recall that $u$ is a neighbor of some vertex from $H$. However, since $u \in  X_1 \cup X_2 \cup \cdots \cup X_t$ and $u \notin X_{\graphStartS{H}}$, the vertex $u$ does no appear in any bag containing a vertex of $H$, so $\cP$ cannot be a path decomposition of $G$. Thus $\graphStartS{H} \leq t$.

Therefore, the bag $X_{t}$ contains at least one vertex from $H$.
Since $S$-branches are vertex-disjoint and $|X_t|\leq k$, we observe that there are at most $k$ $S$-branches $H$, such that  $S \not\subseteq X_{\graphStartS{H}}$.

Now consider an $S$-branch $H$, such that  $S \subseteq X_{\graphStartS{H}}$ and $S \not\subseteq X_{\graphEndS{H}}$.
Let $t'$ be the maximum index such that $S\subseteq X_{t'}\cup\cdots\cup X_l$.

Analogously to the previous case, we observe that $\graphStartS{H} \leq t'$ (by the maximality of $t'$) and $\graphEndS{H} \geq t'$, because $S \not\subseteq X_{\graphEndS{H}}$.
Thus the bag $X_{t'}$ contains at least one vertex from $H$, which shows that the number of $S$-branches $H$, such that $S \subseteq X_{\graphStartS{H}}$ and $S \not\subseteq X_{\graphEndS{H}}$, is at most $k$.
Therefore the total number of $S$-branches, that are not in-branches is at most $2k$, which completes the proof.
\end{proof}

Recall that a non-empty set $S$ is a bottleneck if $|\branches{S}| > 2k$. Thus Lemma \ref{lem:small} implies the following:

\begin{corollary}\label{cor:bottlenecks}
If $\cP=(X_1,\ldots,X_l)$ is a path decomposition and $S$ is a bottleneck, then the following properties hold:
\begin{enumerate}
\item $S$ has at least one in-branch,
\item there is $i$, such that $S \subseteq X_i$,
\item $\card{S} \leq k$. \qed
\end{enumerate}
\end{corollary}
These properties justify the following definition.

\begin{definition} \label{def:interval}
For a bottleneck $S \subseteq V(G)$, let $\botStart{S}{\cP}$ (respectively $\botEnd{S}{\cP})$ be the minimum (respectively maximum) index $i$ such that $i=\graphStart{H}{\cP}$ ($i=\graphEnd{H'}{\cP}$, respectively) for some $S$-branch $H$ (respectively $H'$), which is an in-branch.
\end{definition}
Note that the definition of an in-branch implies that $S \subseteq X_{\botStart{S}{\cP}} \cap X_{\botEnd{S}{\cP}}$.
The interval $I(S, \cP) =[\botStart{S}{\cP},\botEnd{S}{\cP}]$ is called the \emph{interval of} $S$.
Again, we will often write shortly $\botStartS{S}$, $\botEndS{S}$, and $I(S)$, if $\cP$ is clear from the context.

For a bottleneck $S$ we can refine the classification of $S$-branches, which are not in-branches.
We say that an $S$-branch $H$, which is not an in-branch, is
\begin{itemize}
\item a \emph{pre-branch} if $\graphStartS{H} < \botStartS{S}$,
\item a \emph{post-branch} if $\graphStartS{H} \geq  \botStartS{S}$ and $\graphEndS{H} > \botEndS{S}$.
\end{itemize}
By $\branchesin{x}{S,\cP}$ we denote the set of all $x$-branches for $S$, where $x \in \{\preB,\inB,\postB\}$.
Again, if $\cP$ is clear from the context, we will write $\branchesin{x}{S}$ instead of $\branchesin{x}{S,\cP}$.
By $\cC(S)$ we denote the set of all connected components of $G - S$, that are not $S$-branches.

\paragraph*{Example.} \textit{The graph $G$ in Figure~\ref{fig:preinpost} illustrates the above concepts.
The sequence $X_1,X_2,\ldots,X_{16}$ forms a connected path decomposition of $G$.
The only bottleneck set $S$ consists of two vertices denoted by circles. (In this example we take $k=3$.) There are $13$ $S$-branches: one pre-branch ($G[X_1 \cup X_2 \setminus S]$), eleven in-branches ($G[X_i \setminus S]$ for $i=4,5,\ldots,14$), and one post-branch $(G[X_{15} \cup X_{16} \setminus S]$).
The interval of $S$ is equal to $I(S,\cP) = [4, 14]$ and the component $G[X_3 \setminus S]$ is not an $S$-branch, because some vertices in $S$ do not have a neighbor in this component.}

\begin{figure}[htb]
\begin{center}
\centering
  \includegraphics[scale=0.9]{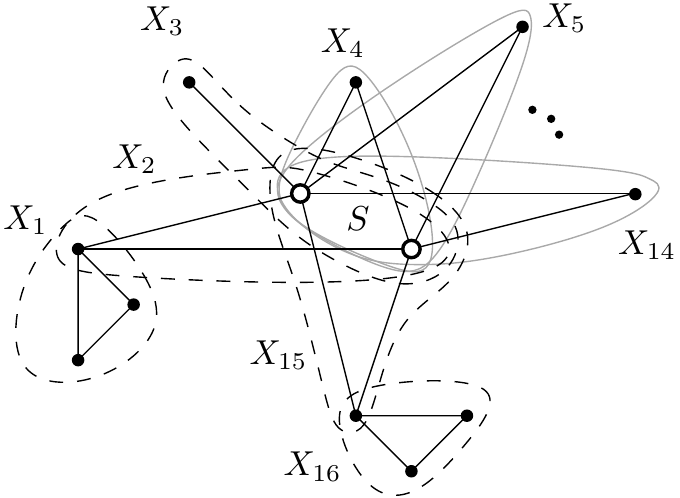}
	\end{center}
\caption{An illustration of two-vertex bottleneck set $S$ and the corresponding $S$-branches.}
\label{fig:preinpost}
\end{figure}

For a subgraph $H$ of $G$, we say that $H$ \emph{waits in} an interval $[i,j]$ of $\cP$ if
\[V(H)\cap X_i=V(H)\cap X_{i+1}=\cdots=V(H)\cap X_j.\]
We say that a path decomposition $\cP$ is $S$-\emph{structured} if each $C$ that is in $C(\cS)$ or is a post- or a pre-branch of $\cS$ waits in $I(S,\cP)$. 
The main technical tool in our approach is the following result concerning the structure of $\cI$-connected path decompositions. The proof of the lemma is provided after giving several technical facts that we need.

\begin{lemma}\label{lem:main_transformation}
If there exists an $\cI$-connected path decomposition $\cP$, then there is also an $\cI$-connected path decomposition $\cP'$ of width at most $\width{\cP}$ such that $\cP'$ is $S$-structured for every bottleneck $S$.
\end{lemma}

We define $\cMin[S,\cP]$ as the minimum index $i \in I(S,\cP)$, for which the size of the set
\[X_{i} \cap \left( \bigcup_{H \in \branchesin{\preB}{S} \cup \branchesin{\postB}{S} \cup \cC(S)} V(H) \right)\]
is minimum. Also, set $X^*:=\left( \bigcup_{H \in \branchesin{\preB}{S} \cup \branchesin{\postB}{S} \cup \cC(S)} V(H) \right) \cap X_{\cMin[S,\cP]}$ to be this minimum-size set.

Let $\Gamma$ be the set of all $\cI$-connected path decompositions of $G$ and $\cS$ be the set of all bottlenecks of $G$.
We will define a function $F : \cS \times \Gamma \rightarrow \Gamma$, which for given $S \in \cS$ and $\cP \in \Gamma$ transforms $\cP$ into an $S$-structured $\cI$-connected path decomposition of width at most $\width{\cP}$.
For simplicity of notation, from now on $t_1 = \botStart{S}{\cP}$, $t_2 = \botEnd{S}{\cP}$, and $\cMin = \cMin[S,\cP]$, whenever $S$ and $\cP$ are clear from the context.
Let $\branchesin{\inB}{S,\cP} = \{H_1,\ldots, H_{l_{\inB}}\}$, where the in-branches are ordered according to their first occurrences in $\cP$, that is, $\graphStartS{H_1} \leq \cdots \leq \graphStartS{H_{l_{\inB}}}$.
Let
\[d= \sum\limits_{H \in \branchesin{\inB}{S}} \left(\graphEndS{H} - \graphStartS{H} + 1\right).\]
In other words, $d$ is the sum of lengths of intervals in which the in-branches $H$ are contained in the path decomposition $\cP = (X_1,X_2,\ldots,X_l)$.

For each in-branch $H_i$, for $i\in\{1,\ldots,l_{\inB}\}$, we define the following sequence:
\[\cP_i: = \left(X_{\graphStartS{H_i}}\cap V(H_i),\ldots, X_{\graphEndS{H_i}}\cap V(H_i)\right).\]
Since $\cP$ is an $\cI$-connected path decomposition of $G$, it is straightforward to observe that $\cP_i$ is a $N_G(S)\cap V(H_i)$-connected path decomposition of $H_i$. We will denote the elements of $\cP_i$ by
\[\cP_i=\left(X_1^i,\ldots, X^i_{\graphEndS{H_i} - \graphStartS{H_i} + 1} \right).\]

Then, let us define a sequence $\cP^*$ as follows:
\[\cP^* := \prod_{i=1}^{l_{\inB}} \; \cP_i,\]
where $\prod$ denotes the concatenation of sequences. Observe that the length of $\cP^*$ is exactly $d$. We will denote the elements of $\cP^*$ by $\cP^*=\left(X_{1}^*,\ldots,X^*_{d}\right)$.

Define $B^{\inB}$ to be the set of vertices of the in-branches of $S$, i.e., $B^{\inB} := \bigcup_{H \in \branchesin{\inB}{S,\cP}} V(H)$. Now, we define a path decomposition $F(S,\cP) = \left(X'_1,\ldots,X'_{l + d + 1}\right)$ as follows:
\begin{numcases}{X_i'=}
X_i	& for $i \in \{1,\ldots, t_1 - 1\}$; \label{eq:prefix}\\
X_i \setminus B^{\inB}		& for $i\in\{t_1,\ldots, \cMin - 1\}$; \label{eq:tillcmin}\\
X_{\cMin}\setminus B^{\inB} & for $i = \cMin$; \label{eq:minimal1}\\
X^* \cup X_{i - \cMin}^*\cup S		 		& for $i\in\{\cMin + 1,\ldots, \cMin + d\}$; \label{eq:waiting}\\
X_{\cMin}\setminus B^{\inB} & for $i = \cMin + d + 1$;\label{eq:minimal2}\\
X_{i - d - 1}\setminus B^{\inB}	& for $i\in\{\cMin + d + 2,\ldots, t_2 + d + 1\}$; \label{eq:aftercmin}\\
X_{i - d - 1}										& $i\in\{t_2 + d + 2 ,\ldots, l + d + 1\}$. \label{eq:suffix}
\end{numcases}
Observe that $F(S,\cP)$ is obtained from $\cP$ by a modification of the interval $[t_1,t_2]$ of $\cP$. The prefix $(X_1,X_2,\ldots,X_{t_1-1})$ and the suffix $(X_{t_2+1},X_{t_2+2},\ldots,X_{l})$ of $\cP$ are just copied into $F(S,\cP)$ without any changes (see conditions \eqref{eq:prefix} and \eqref{eq:suffix}).
All components of $G-S$, apart from the in-branches, are covered by bags on positions up to $\cMin-1$, see~\eqref{eq:tillcmin}, and after position $\cMin+d + 1$, see~\eqref{eq:aftercmin}, and they wait for $d + 2$ steps in the interval $[\cMin,\cMin+d + 1]$ of $F(S,\cP)$ --- see~\eqref{eq:minimal1}--\eqref{eq:minimal2}.
The interval $[\cMin +1,\cMin+d]$ is used in~\eqref{eq:waiting} to cover all in-branches, one by one, in the order of their appearance in $\cP$. Additional two bags without any vertices of in-branches are added on positions $\cMin$ and $\cMin + d + 1$ to ensure that for any other bottleneck $S'$, such that $I(S,\cP) \subsetneq I(S',\cP)$, we have that $\cMin(S',\cP) \not\in I(S,\cP)$, see~\eqref{eq:minimal1} and~\eqref{eq:minimal2}. See Figure~\ref{fig:trans1} that illustrates the conversion from $\cP$ to $F(S,\cP)$ for a bottleneck $S$.
Notice that the interval of $S$ in $\cP'$ is given by $\botStart{S}{\cP'} = \cMin + 1$  and $\botEnd{S}{\cP'} = \cMin + d$, and it is possibly different than  $[t_1,t_2]$.
In the next lemma we show that $F(S,\cP)$ has all necessary properties.

  \begin{figure}[htb]
\centering
\begin{subfigure}{\textwidth}
  \centering
  \includegraphics[width=0.55\textwidth]{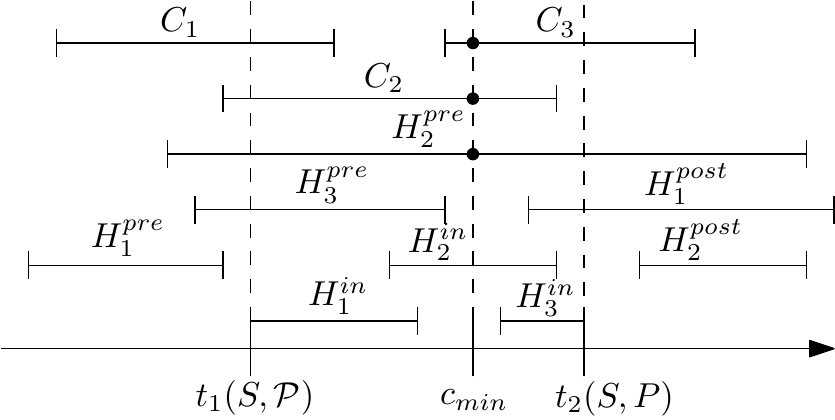}
  \caption{$S$-branches for a path decomposition $\cP$ before transformation.}
  \label{fig:trans1b}
\end{subfigure}
\begin{subfigure}{\textwidth}
  \centering
  \includegraphics[width=0.85\textwidth]{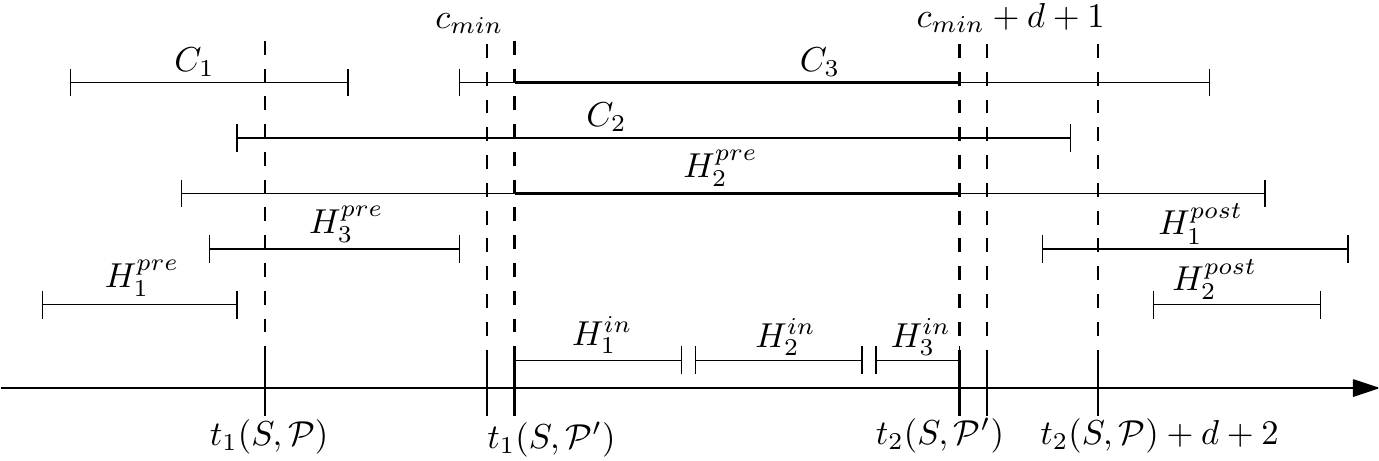}
  \caption{New bags are inserted into $\cP$ in which all components apart from in-branches wait in the interval of $S$ in $\cP'$,  i.e., $C_2,C_3,H_2^{\preB}$ wait in $[\botStart{S}{\cP'}, \botEnd{S}{\cP'}]$.}
  \label{fig:trans1a}
\end{subfigure}
\caption{Illustration of the conversion from $\cP$ to $\cP' = F(S,\cP)$ for a bottleneck $S$. For simplification, $S$ is assumed to only have $8$ $S$-branches and $3$ components, which are denoted by letters $H$ and $C$, respectively, with appropriate indices.}
\label{fig:trans1}
\end{figure}

\begin{lemma}\label{lem:reorganizing}
For any $\cI$-connected path decomposition $\cP$ and a bottleneck $S$, $\cP' = F(S,\cP)$ is an $\cI$-connected path decomposition with $\width{\cP'} \leq \width{\cP}$.
\end{lemma}
\begin{proof}
Let us recall the notation: $\cP := (X_1,\ldots, X_l)$, $\cP' = F(S,\cP) = (X_1',\ldots, X'_{l'})$, $t_1 = \botStart{S}{\cP}$ and $t_2 = \botEnd{S}{\cP}$, $B^{\inB} := \bigcup_{H \in \branchesin{\inB}{S,\cP}} V(H)$.
Moreover, define
\[B^{\textup{not-in}}:=\left( \bigcup_{H \in \branchesin{\preB}{S,\cP} \cup \branchesin{\postB}{S,\cP} \cup \cC(S)} V(H) \right).\]
Also, recall that $X'_i = X_i,\ 1 \leq i < t_1$ and $X'_i = X_{i-d-1}$, for all $t_2 + d + 2 \leq i \leq l'$.

First, we want to show that $\cP'$ satisfies conditions in Definition \ref{def:path_dec}.

Let $\{u,v\}$ be an edge of $G$. Since $\cP$ is a path decomposition, $u,v \in X_i$ for some $i$.
If $i < t_1$, then $u,v \in X'_i=X_i$. If $i > t_2$, then $u,v \in X'_{i+d+1} =X_i$.

So suppose that $u,v \in X_i$ for $i \in [t_1,t_2]$.
Note that this means that $u,v \in B^{\inB} \cup S$ or $u,v \in B^{\textup{not-in}} \cup S$. 
If $u,v \in B^{\textup{not-in}} \cup S$, then we have
\begin{enumerate}[label=(\roman*)]
\item $u,v \in X'_i$, if $i \leq \cMin$;
\item $u,v \in X'_{i+d+1}$, otherwise.
\end{enumerate}

Finally, consider the case that, say, $u \in B^{\inB}$ and $v \in B^{\inB} \cup S$ (if both $u,v \in S$, then we are at the previous case). This means that $u$ is a vertex of some in-branch $H_s \in \branchesin{\inB}{S,\cP}$, so it appears in some bag of $\cP_s$ and thus of $\cP^*$. This implies that $u,v \in X'_j$ for some $j \in [\cMin + 1,\cMin+d]$. This implies that $\cP'$ satisfies conditions \ref{it:path1} and \ref{it:path2} from Definition \ref{def:path_dec}.

Now let us verify that the condition \ref{it:path3} is also satisfied, i.e., for every vertex $v$ and indices $i < s < j$, such that $v \in X'_i \cap X'_j$ we have $v \in X'_s$. Clearly the condition is satisfied for every $v$ and  $j < t_1$ or $i > t_2$, since these parts of $\cP'$ are just copied from $\cP$ without any modifications.
The situation is very similar if $i \leq t_2$ and $j \geq t_1$ and $v \in \left ( S \cup B^{\textup{not-in}} \right ) \setminus X^*$. If $v \in X^*$, then $v$ is included in all bags $X'_s$ for $s \in [\cMin ,\cMin+d +1]$, so the condition \ref{it:path3} follows from the correctness of $\cP$. Finally, if $v$ is a vertex of some $H_a \in \branchesin{\inB}{S,\cP}$, then the condition \ref{it:path3} follows from the property \ref{it:path3} that holds for the path decomposition $P_a$.

Next thing to show is that $\width{\cP'} \leq \width{\cP}$.
Let $c_1 = \card{X^*}$ and $c_2 = \max\limits_{i=t_1,\ldots,t_2} \card{B^{\inB} \cap X_i}$, and let $k=\width{\cP} +1 = \max\limits_{i=1,\ldots,l} \card{X_i}$. Observe that each $X'_i$ for $i  \notin \{t_1,\ldots,t_2+d + 1\}$ is an exact copy of some $X_j$, so $\card{X'_i} \leq k$. Moreover, each $X'_i$ for $i\in \{t_1,\ldots,\cMin\} \cup \{\cMin+d + 1,\ldots,t_2 + d + 1\}$ was obtained from some $X_j$ by removal of vertices of $B^{\inB}$, so again we have $\card{X_i} \leq k$.
Finally, for $i \in \{\cMin + 1,\ldots,\cMin + d\}$ we have 
\[
\card{X'_i} = \card{S} + \card{X^*} + \card{X_{i - \cMin}^*} \leq \card{S} + c_1 + c_2.
\]
However, by the definition of $\cP^*$, we observe that $\card{S} + c_1 + c_2 \leq \card{X_j}$ for some $j \in \{t_1,\ldots,t_2\}$. Therefore, $\width{\cP'} \leq \width{\cP}$.

Finally, recall that $\cP$ is $\cI$-connected, i.e., $X_1$ is a one-element subset of $\cI$. Observe that $X_{\botStart{S}{\cP}}$ contains whole $S$ and at least one vertex from an $S$-branch (an in-branch, to be more specific). Since $S \neq \emptyset$, we conclude that $X_{\botStart{S}{\cP}}$ cannot be $X_1$, so $X'_1 = X_1$ and therefore $\cP'$ is $\cI$-connected too.
\end{proof}

Observe that every connected component $H$ of $G-S$ is either contained in $I(S, \cP')$ (which means that $H$ is an in-branch) or waits in $I(S, \cP')$ (for all other $H$).

\begin{observation}\label{obs:structured}
For any $\cP$ and bottleneck $S$, the path decomposition $F(S,\cP)$ is $S$-structured. \qed
\end{observation}

Now we want to define a series of transformations, which start at an arbitrary $\cI$-connected path decomposition $\cP$ and transform it into an $\cI$-connected path decomposition with no larger width, which is $S$-structured for every $S \in \cS$. For this, we will apply the $F$-transformations for all bottlenecks. In order to do this we need some technical lemmas about the structure of bottlenecks and their branches.

\begin{lemma}\label{lem:subsets_bottlenecks}
Let $S$ and $S'$ be two bottlenecks, such that $S'\subsetneq S$. There exists an $S'$-branch $H$ such that $\bigcup_{H' \in \branches{S}} V(H') \cup S  \setminus S' \subseteq V(H)$ and every $S'$-branch $H'$, different than $H$, is a non-branch connected component of $G-S$. 
\end{lemma}
\begin{proof}
Let $S, S' \in \cS$ such that $S'\subsetneq S$. Clearly $S$ intersects some connected component of $G-S'$.
Since every $S$-branch $H''$ is connected and $S\setminus S' \subset N_G(V(H''))$, we observe that two connected components of $G-S'$ can not be distinctive, i.e., there exists a connected component $H$ of $G-S'$ such that $\bigcup_{H' \in \branches{S}} V(H') \cup S  \setminus S' \subseteq V(H)$.

To see that $H$ is an $S'$-branch, consider a vertex $s' \in S'$. By assumption, $S'\subseteq S$ and hence $s'$ is also a vertex of $S$.
Thus, $s'$ has a neighbor in every $S$-branch $H'$, and thus in $H$. See Figure~\ref{fig:case1} for an illustration.

Now consider an $S'$-branch $H'' \neq H$. If every vertex of $S \setminus S'$ is adjacent to some vertex of $H''$,  then $H''$ is an $S$-branch, a contradiction. Thus, $H''$ is a connected component of $G-S$, which is not an $S$-branch. 
\end{proof}

  \begin{figure}[htb]
  \centering
  \includegraphics[width=0.5\textwidth]{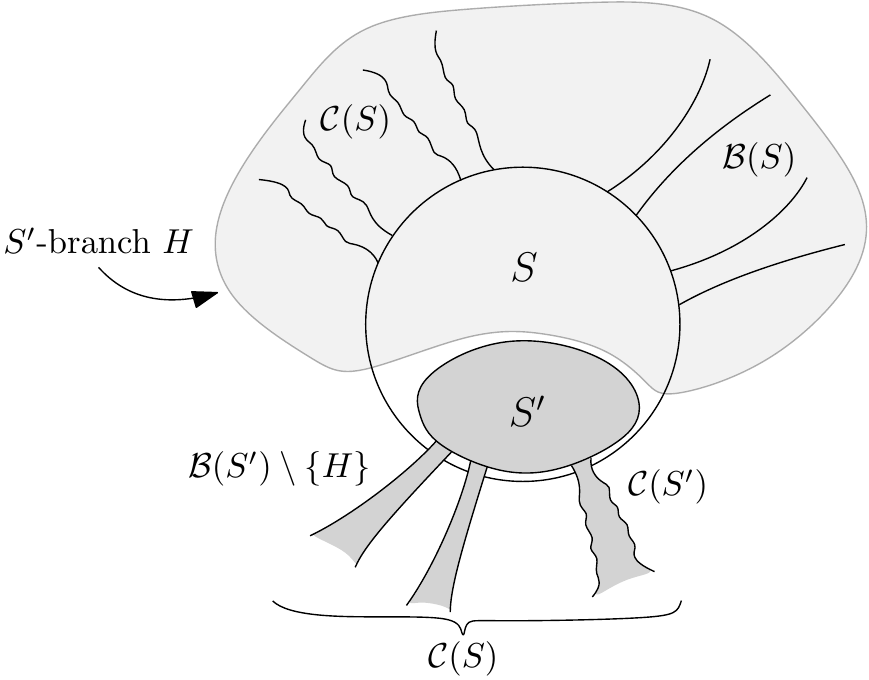}
  \caption{Illustration for Lemma~\ref{lem:subsets_bottlenecks}; $S$ and $S'$ are bottlenecks, such that  $S' \subsetneq S$. All $S$-branches are subgraphs of an $S'$-branch $H$ and every $S'$-branch, different than $H$, is a connected component of $G-S$, which is not an $S$-branch.  }
  \label{fig:case1}
\end{figure}

\begin{lemma}\label{lem:one_component}
For any two bottlenecks $S,S'$, if $S'\nsubseteq S$, then there exists exactly one connected component $H$ of $G-S$ such that $S'\subseteq S\cup V(H).$
\end{lemma}
\begin{proof}
Let $S, S' \in \cS$ such that $S'\nsubseteq S$. Clearly $S'$ intersects some connected component of $G-S$. Suppose that $S'$ has a non-empty intersection with two connected components $H$ and $H'$ of $G-S$.
Thus, since every $S'$-branch  $H''$ is connected and $N_G(V(H''))=S'$, we observe that $H''$ contains a vertex of $S$ (otherwise $H,H'$ would not be two distinct components). Since $S'$-branches are vertex-disjoint, this implies that 
the number of such $S'$-branches is at most $\card{S}$, which is in turn at most $k$ by Corollary \ref{cor:bottlenecks}. However, this contradicts the assumption that $S'$ is a bottleneck.
\end{proof}

The next remark is a straightforward consequence of Lemma~\ref{lem:subsets_bottlenecks} and Lemma~\ref{lem:one_component}.

\begin{remark}\label{rem:all_branches}
Let $S,S'$ be bottlenecks such that $S'\nsubseteq S$ and $S \nsubseteq S'$. Let $H$ be the connected component of $G-S$, such that $S'\subseteq S\cup V(H)$.
There exists exactly one connected component $C$ of $G-S'$ such that $\bigcup_{H' \in \branches{S} \setminus H} V(H') \cup S \setminus S' \subseteq V(C)$.
Moreover, all $S'$-branches but possibly  $C$ are subgraphs of $H$.
\qed
\end{remark}
 
See Figure~\ref{fig:case2-3} for the illustration for Lemma~\ref{lem:one_component} and Remark~\ref{rem:all_branches}. 

\begin{figure}[htb]
\centering
\begin{subfigure}{0.5\textwidth}
  \centering
  \includegraphics[width=0.83\textwidth]{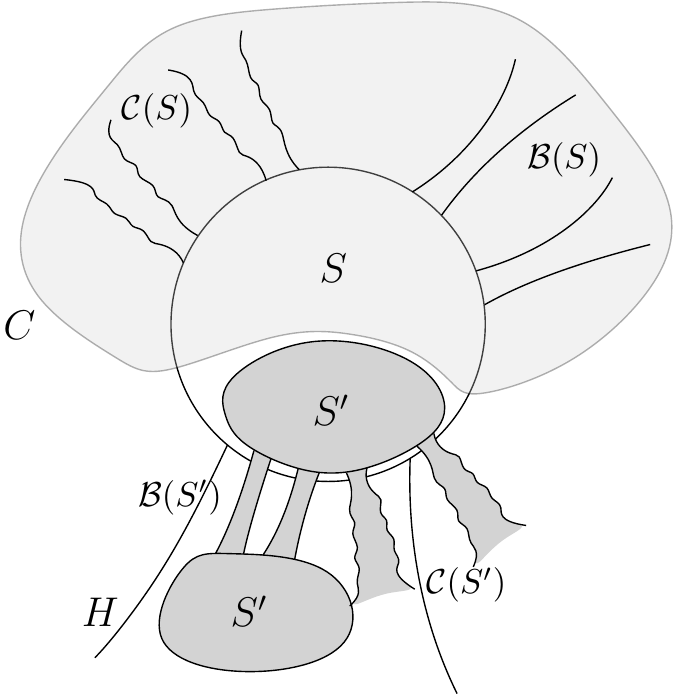}
  \caption{$S \setminus S' \neq \emptyset$ and $S \cap S' \neq \emptyset$.}
  \label{fig:case2}
\end{subfigure}~~
\begin{subfigure}{0.5\textwidth}
  \centering
  \includegraphics[width=0.83\textwidth]{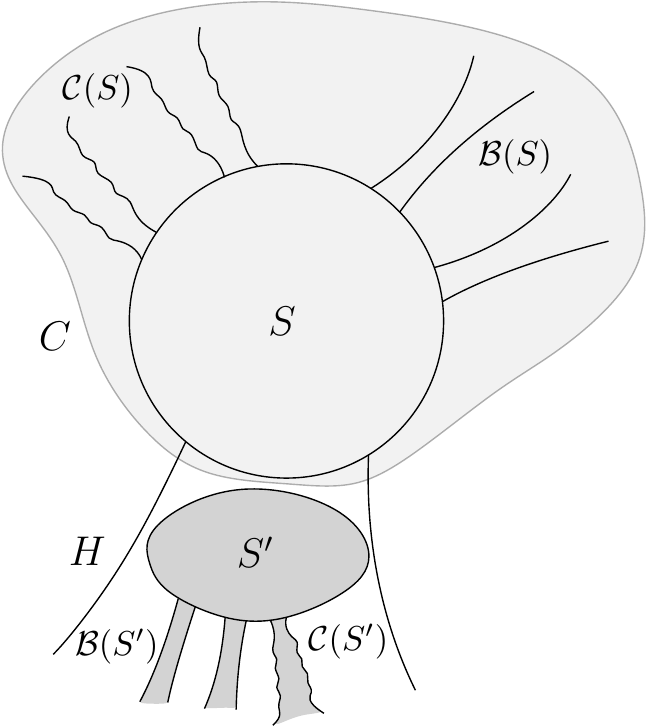}
  \caption{$S \cap S' = \emptyset$. }
  \label{fig:case3}
\end{subfigure}
\caption{Illustration of two cases for Lemma~\ref{lem:one_component} and Remark~\ref{rem:all_branches}; $S$ and $S'$ are bottlenecks. All $S$-branches (apart from $H$, which is a connected component of $G-S$, which may or may be not an $S$-branch) are subgraphs of $C$ and $S' \subseteq S \cup V(H)$.}
\label{fig:case2-3}
\end{figure}

We say two bottlenecks $S, S'$ are \emph{well-nested} in $\cP$ if 
\begin{enumerate}[label=\textup{(\roman*)}]
\item $I(S',\cP) \subsetneq I(S,\cP)$ or
\item $I(S,\cP) \subseteq I(S',\cP)$ or 
\item $I(S,\cP) \cap I(S',\cP) = \emptyset$.
\end{enumerate}
Observe that the ordering of $S,S'$ in the definition above matters.

\begin{lemma} \label{lem:structured}
For any $\cI$-connected path decomposition $\cP$ and bottlenecks $S,S'$, if $\cP$ is $S'$-structured, then $S$ and $S'$ are well-nested in $\cP$.
Moreover, if $S \subsetneq S'$, then either $I(S,\cP) \cap I(S',\cP) = \emptyset$ or $I(S',\cP) \subseteq I(S,\cP)$.
\end{lemma}
\begin{proof}
Let $\cP = (X_1,\ldots,X_l)$ be an $\cI$-connected path decomposition of $G$ and let $S, S' \in \cS$. Suppose $\cP$ is $S'$-structured.
Let us assume that $I(S',\cP) \cap I(S,\cP) \neq \emptyset$, we will show that either $I(S',\cP) \subsetneq I(S,\cP)$ or $I(S,\cP) \subseteq I(S',\cP)$.

\paragraph*{Case A: $S'\subsetneq S$.}
By Lemma~\ref{lem:subsets_bottlenecks}, there exists an $S'$-branch $H$ such that for every $S$-branch $H'$ it holds that $S \setminus S' \cup V(H')\subseteq V(H)$, and thus $[\graphStartS{H'},\graphEndS{H'}] \subseteq [\graphStartS{H},\graphEndS{H}]$. Recall that $\botStartS{S} = \graphStartS{H_1'}$ and $\botEndS{S} = \graphEndS{H_2'}$ for some in-branches $H_1',H_2' \in \branches{S}$, so $I(S,\cP) = [\botStartS{S}, \botEndS{S}] \subseteq [\graphStartS{H},\graphEndS{H}]$.
Consider two subcases.
\subparagraph*{Subcase A1: $\graphStartS{H} \in I(S',\cP)$.} 
Since $\cP$ is $S'$-structured, we observe that $H$ is an in-branch for $S'$, which implies that $I(S,\cP) \subseteq [\graphStartS{H},\graphEndS{H} ] \subseteq I(S',\cP)$.
\subparagraph*{Subcase A2: $\graphStartS{H} \notin I(S',\cP)$.}
First, observe that if $\graphStartS{H} > \botEndS{S'}$, then $\botStartS{S} > \botEndS{S'}$ and thus $I(S,\cP) \cap I(S',\cP) = \emptyset$, which contradicts our assumption. Analogously, if $\graphEndS{H} < \botStartS{S'}$, we again obtain that $I(S,\cP) \cap I(S',\cP) = \emptyset$.
Therefore assume that $\graphStartS{H} < \botStartS{S'} \leq \graphEndS{H}$. Observe that this implies that $H$ is a pre-branch for $S'$ and, since $\cP$ is $S'$-structured, $H$ waits in $I(S',\cP)$. In particular $(S \setminus S') \cap X_{\botStartS{S'}} = (S \setminus S') \cap X_{\botStartS{S'}+1} = \ldots = (S \setminus S') \cap X_{\botEndS{S'}}$. Since $I(S',\cP) \cap I(S,\cP) \neq \emptyset$, it is necessary that $\botStartS{S} \leq \botStartS{S'}$ (otherwise $\botStartS{S} > \botEndS{S'}$).
 Thus, $\botStartS{S} \leq \botStartS{S'} \leq \botEndS{S'} \leq \botEndS{S}$ (recall $H$ waits in $I(S',\cP)$).
Summing up, if $S' \subsetneq S$, then either $I(S',\cP) \subsetneq I(S,\cP)$ or $I(S,\cP) \subseteq I(S',\cP)$, which completes the proof for this case.

\paragraph*{Case B: $S'\nsubseteq S$.}
By Lemma~\ref{lem:one_component}, there exists exactly one connected component $H$ of $G-S$ such that $S'\subseteq S\cup V(H)$. 
Since $S'\nsubseteq S$, we observe that $V(H)\cap S'\neq\emptyset$.

If $H$ is an $S$-branch that is an in-branch in $\cP$, then
\[I(S',\cP) \subseteq [\graphStart{H}{\cP},\graphEnd{H}{\cP}] \subseteq I(S, \cP).\]
We observe that $I(S',\cP) \subseteq I(S, \cP)$ is equivalent to $I(S',\cP) \subsetneq I(S, \cP)$ or $I(S',\cP) = I(S, \cP)$, thus $S$ and $S'$ are well-nested in $\cP$. So assume that $H$ is a connected component of $G-S$, that is not an in-branch for $S$ (it may still be a pre- or a post-branch). We consider now two subcases.

\subparagraph*{Subcase B1: $S \subsetneq S'$.}
By Lemma~\ref{lem:subsets_bottlenecks}, all $S$-branches possibly except for $H$ are not $S'$-branches and all $S'$-branches are subgraphs of $H$.

Because $\cP$ is $S'$-structured, every $S$-branch but possibly $H$ waits in $I(S',\cP)$.
In particular, every in-branch $H''$ for $S$ does wait in $I(S',\cP)$.
Thus, for every such an in-branch $H''$ it holds that $I(S',\cP)\subseteq [\graphStartS{H''},\graphEndS{H''}]$ or $I(S',\cP) \cap [\graphStartS{H''},\graphEndS{H''}]=\emptyset$.
Note that the second condition implies that $I(S',\cP)\cap I(S,\cP)=\emptyset$, which contradicts our assumption.
Therefore  we obtain that $I(S',\cP)\subseteq [\graphStartS{H''},\graphEndS{H''}] \subseteq I(S,\cP)$.
Note that this shows the second claim of the lemma.

\subparagraph*{Subcase B2: $S \setminus S' \neq \emptyset$.} Let $C$ be the connected component of $G-S'$, for which it holds $\bigcup_{H' \in \branches{S}\setminus H} V(H') \cup S \setminus S' \subseteq V(C)$, whose existence is guaranteed by Remark~\ref{rem:all_branches}. 

If $C$ is an in-branch for $S'$, we observe that $I(S,\cP) \subseteq [\graphStartS{C}, \graphEndS{C}] \subseteq I(S',\cP)$, because $\cP$ is $S'$-structured.
On the other hand, if $C$ is not an in-branch for $S'$, then all in-branches of $S$ wait in $I(S',\cP)$. This is because all subgraphs but in-branches of $S'$ wait in $I(S',\cP)$, and every in-branch for $S$ is vertex-disjoint with every in-branch for $H$, since they are all contained in $V(C) \cup S'$. Thus, since $I(S,\cP) \cap I(S',\cP) \neq \emptyset$, we conclude that $I(S',\cP) \subseteq I(S,\cP)$, which completes the proof.
\end{proof}

In the next lemma, we show that we can apply a series of $F$-transformations, one for each bottleneck, so that the structure obtained in previous $F$-transformations is not 'destroyed' during the subsequent $F$-transformations.

\begin{lemma}\label{lem:final_structured}
Let $S,S'$ be bottlenecks and let $\cP$ be an $S$-structured $\cI$-connected path decomposition. Let $\cP' = F(S',\cP)$.
\begin{enumerate}
\item If $I(S',\cP) \subseteq I(S,\cP)$, then $\cP'$ is $S$-structured and $I(S',\cP') \subseteq I(S,\cP')$. \label{case:nested}
\item If $\botEnd{S}{\cP} < \botStart{S'}{\cP}$, then $\cP'$ is $S$-structured and $\botEnd{S}{\cP'} < \botStart{S'}{\cP'}$. \label{case:disjoint1}
\item If $\botEnd{S'}{\cP} < \botStart{S}{\cP}$, then $\cP'$ is $S$-structured and $\botEnd{S'}{\cP'} < \botStart{S}{\cP'}$. \label{case:disjoint2}
\end{enumerate}
\end{lemma}
\begin{proof}
Let $\cP = (X_1,\ldots,X_l)$ be an $\cI$-connected path decomposition of $G$, let $S, S' \in \cS$ and $\cP' = F(S',\cP)$. Moreover, assume that $\cP$ is $S$-structured. 

First, let us prove case \ref{case:nested}, i.e. we assume that $I(S',\cP) \subseteq I(S,\cP)$.
Observe that by Lemma \ref{lem:structured} we obtain that $S' \nsubseteq S$.
By Lemma~\ref{lem:one_component}, there exists exactly one connected component $H$ of $G-S$ such that $S' \subseteq S \cup V(H)$. We observe that $H$ is an in-branch for $S$, otherwise $H$ would wait in $I(S,\cP)$, which contradicts the assumption that $I(S',\cP) \subseteq I(S,\cP)$ and that $\cP$ is $S$-structured. We are going to show now that every $S'$-branch $H'$, which is an in-branch, is a subgraph of $H$. If $S \subsetneq S'$, then we obtain it immediately from Lemma~\ref{lem:subsets_bottlenecks}, so let $S \setminus S' \neq \emptyset$. Let $C$ be a connected component of $G-S'$, such that $\bigcup_{H' \in \branches{S}\setminus H} V(H') \cup S \setminus S' \subseteq V(C)$, whose existence is guaranteed by Remark~\ref{rem:all_branches}. Recall that $C$ might  or might not be an $S'$-branch, but for sure it is not an in-branch for $S'$, because then $I(S,\cP) \subsetneq I(S',\cP)$, which is a contradiction. Thus, by Remark~\ref{rem:all_branches}, every in-branch for $S'$ is a subgraph of $H$.

Because every in-branch for $S'$ is a subgraph of $H$, we conclude that $I(S',\cP) \subseteq [\graphStart{H}{\cP}, \graphEnd{H}{\cP}]$ and the only changes made by the transformation $F(S',\cP)$ concern the vertices of $H \in \branchesin{\inB}{S}$. Every connected component of $G-S'$, apart from in-branches (for $S'$), waits in $I(S',\cP')$, so $F(S',\cP)$ is $S$-structured and $I(S',\cP') \subseteq [\graphStart{H}{\cP'}, \graphEnd{H}{\cP'}]\subseteq I(S,\cP')$.

To see that cases \ref{case:disjoint1} and \ref{case:disjoint2} hold as well, notice that the prefix $(X_1,X_2,\ldots,X_{\botStart{S'}{\cP}-1})$ and suffix $(X_{\botEnd{S'}{\cP}+1}, X_{\botEnd{S'}{\cP}+2},\ldots,X_l)$ of $\cP$ are just copied into $\cP'$ without any changes.
\end{proof}

In the following lemma it is crucial that the path decomposition $\cP$ is not only $S$-structured but has been obtained by by applying the transformation described in \eqref{eq:prefix}-\eqref{eq:suffix} to $\cP_0$. In particular, the bags added in \eqref{eq:minimal1} and~\eqref{eq:minimal2} will play a crucial role in ensuring that the path decomposition returned by $F(S',\cP)$ remains $S$-structured.
\begin{lemma}\label{lem:final_structured2}
Let $S,S'$ be bottlenecks and let $\cP_0$ be any $\cI$-connected path decomposition. Let $\cP = F(S, \cP_0)$ and $\cP' = F(S',\cP)$.
If $I(S,\cP) \subsetneq I(S',\cP)$, then $\cP'$ is $S$-structured and $I(S,\cP') \subseteq I(S',\cP')$ or $I(S,\cP') \cap I(S',\cP') = \emptyset$.
\end{lemma}
\begin{proof}
Let $\cP = (X_1,\ldots,X_l)$ be an $\cI$-connected path decomposition of $G$, let $S, S' \in \cS$ and $\cP' = F(S',\cP)$. Moreover, assume that $\cP = F(S, \cP_0)$ for some $\cI$-connected path decomposition $\cP_0$. In particular, this implies that $\cP$ is $S$-structured.
Finally, assume that $I(S,\cP) \subsetneq I(S',\cP)$.

\paragraph*{Case A: $S' \subsetneq S$.} By Lemma~\ref{lem:subsets_bottlenecks} there exists an $S'$-branch $H$ such that $\bigcup_{H' \in \branches{S}} V(H') \cup S  \setminus S' \subseteq V(H)$.
If $H$ is an in-branch for $S'$ then $I(S,\cP) \subseteq [\graphStart{H}{\cP}, \graphEnd{H}{\cP}] \subsetneq I(S',\cP)$.
However, recall that $F$-transformation applied to $S'$ and $\cP$ does not change the structure of the bags restricted to $H$ (or any other in-branch of $S'$). Therefore we conclude that $I(S,\cP') \subseteq [\graphStart{H}{\cP'}, \graphEnd{H}{\cP'}] \subsetneq I(S',\cP')$ and $\cP'$ is $S$-structured.

Now assume that $H$ is a pre-branch or a post-branch for $S'$. From the construction of $\cP$ we have that $\cMin{(S',\cP)} \not\in I(S,\cP)$.
Because every in-branch $H'$ for $S'$ either waits in $I(S,\cP)$ or $[\graphStart{H'}{\cP}, \graphEnd{H'}{\cP}] \cap I(S,\cP) = \emptyset$, we obtain that $I(S,\cP') \cap I(S',\cP') = \emptyset$ and $\cP'$ is $S$-structured. 

\paragraph*{Case B: $S' \nsubseteq S$ and $S \nsubseteq S'$.}
By Lemma~\ref{lem:one_component} there exists exactly one connected component $H$ of $G-S$ such that $S' \subseteq S \cup V(H)$. Because $V(H) \cap S' \neq \emptyset$ we have that $I(S',\cP) \subseteq [\graphStart{H}{\cP}, \graphEnd{H}{\cP}]$. We observe that $H$ cannot be an in-branch for $S$, otherwise $I(S',\cP) \subseteq [\graphStart{H}{\cP}, \graphEnd{H}{\cP}] \subseteq I(S, \cP)$, which contradicts the assumption that $I(S,\cP) \subsetneq I(S',\cP)$.
Thus, $H \not\in \branchesin{\inB}{S}$.

From the facts that $\cP$ is $S$-structured and $I(S',\cP) \subseteq [\graphStart{H}{\cP}, \graphEnd{H}{\cP}]$, we observe that $H$ waits in $I(S,\cP)$.
By Remark~\ref{rem:all_branches} there exists a connected component $C$ of $G-S'$ such that $\bigcup_{H' \in \branches{S} \setminus H} V(H') \cup S \setminus S' \subseteq V(C)$. Because $H$ is not an in-branch for $S$, all in-branches of $S$ are subgraphs of $C$. If $C$ is an in-branch for $S'$ then $I(S,\cP) \subseteq [\graphStart{C}{\cP}, \graphEnd{C}{\cP}] \subsetneq I(S',\cP)$. Because $F$-transformation does not change bags inside one in-branch, we have that $I(S,\cP') \subseteq [\graphStart{C}{\cP'}, \graphEnd{C}{\cP'}] \subsetneq I(S',\cP')$ and $\cP'$ is $S$-structured. On the other hand, assume that $C$ is not an in-branch for $S'$. From construction of $\cP$ we have that $\cMin{(S',\cP)} \not\in I(S,\cP)$ Because every in-branch $H'$ for $S'$ either waits in $I(S,\cP)$ or $[\graphStart{H'}{\cP}, \graphEnd{H'}{\cP}] \cap I(S,\cP) = \emptyset$, we obtain that $I(S,\cP') \cap I(S',\cP') = \emptyset$ and $\cP'$ is $S$-structured. 

\paragraph*{Case C: $S \subsetneq S'$.} From Lemma~\ref{lem:subsets_bottlenecks} there exists an $S$-branch $H$ such that $\bigcup_{H' \in \branches{S'}} V(H') \cup S'  \setminus S \subseteq V(H)$, i.e., $I(S',\cP) \subseteq [\graphStart{H}{\cP}, \graphEnd{H}{\cP}]$. We observe that $H$ cannot be an in-branch for $S$, otherwise $I(S',\cP) \subseteq [\graphStart{H}{\cP}, \graphEnd{H}{\cP}] \subseteq I(S, \cP)$, which contradicts the assumption that $I(S,\cP) \subsetneq I(S',\cP)$.

Let then $H$ be a pre-branch or a post-branch for $S$, which means that $H$ waits in $I(S,\cP)$ (notice that it is impossible that $[\graphStart{H}{\cP}, \graphEnd{H}{\cP}] \cap I(S,\cP) = \emptyset$, because then  $ I(S',\cP) \cap I(S,\cP) = \emptyset$). From construction of $\cP$ we have that $\cMin{(S',\cP)} \not\in I(S,\cP)$. Because every in-branch $H'$ for $S'$ either waits in $I(S,\cP)$ or $[\graphStart{H'}{\cP}, \graphEnd{H'}{\cP}] \cap I(S,\cP) = \emptyset$, we obtain that $I(S,\cP') \cap I(S',\cP') = \emptyset$ and $\cP'$ is $S$-structured. 
\end{proof}

Now we are ready to prove  Lemma~\ref{lem:main_transformation}.

\begin{proof}[Proof of Lemma~\ref{lem:main_transformation}.] Let $\cP$ be an $\cI$-connected path decomposition. Let $\cS = \{S_1,S_2,\ldots,S_{n'}\}$ be the set of all bottlenecks.
We define a path decomposition $\cP' := \cP_{n'}$ in the following, recursive way: 
\begin{align*}
\cP_0 &= \cP; \\
\cP_i &= F(S_i,\cP_{i-1}),\ \text{ for } i \in \{1,\ldots,n'\}. 
\end{align*}

We are going to prove now that $\cP_q$ is $S_j$-structured and $S_i, S_j$ are well-nested in $\cP_q$
for every $1 \leq j \leq q \leq n'$ and $1 \leq i \leq n'$. 

Induction on $q$. If $q=1$, then obviously $j=1$ and thesis holds for every $1 \leq i \leq n'$ from Lemma~\ref{lem:structured}.
So assume that $q > 1$ and the claim holds for $q - 1$. 

Let $j \in \{1,\ldots,q - 1\}$. For every $S_j$ 
we have, by the induction assumption, that $S_q, S_j$ are well-nested in $\cP_{q-1}$ and $\cP_{q-1}$ is $S_j$-structured.
By Lemma~\ref{lem:final_structured2}, if $I(S_j,\cP_{q-1}) \subsetneq I(S_q,\cP_{q-1})$, then  $\cP_q = F(S_q,\cP_{q-1})$ is $S_j$-structured.
By Lemma~\ref{lem:final_structured}, if $I(S_q,\cP_{q-1}) \subseteq I(S_j,\cP_{q-1})$  or  $I(S_j,\cP_{q-1}) \cap I(S_q,\cP_{q-1}) = \emptyset$, then  $\cP_q = F(S_q,\cP_{q-1})$ is $S_j$-structured.

Because $\cP_q$ is $S_j$-structured, then from Lemma~\ref{lem:structured} for any $i \in \{1,\ldots,n'\}$ we have that $S_i, S_j$ are well-nested in $\cP_q$.

So $\cP'$ is $S$-structured for every bottleneck $S \in \cS$. 
Note that Lemma~\ref{lem:reorganizing} ensures that $\cP'$ is an $\cI$-connected path decomposition of width at most $\width{\cP}$, which finishes the proof.
\end{proof}

Now we are ready to show the correctness of our algorithm. We will prove it in two steps.

\begin{lemma} \label{lem:induction-dec-totab}
If $G$ has an $\cI$-connected path decomposition of width at most $k-1$, then $Tab[s]=true$ for some state $s$ such that $\cover{s} = V(G)$.
\end{lemma} 
\begin{proof}
Suppose that $G$ has an $\cI$-connected path decomposition of width $k-1$.
By Lemma~\ref{lem:main_transformation}, there exists an $\cI$-connected path decomposition $\cP=(X_1,\ldots,X_l)$ that has width $k-1$ and is $S$-structured for each bottleneck $S$.

By Lemma~\ref{lem:structured}, the set $\cS$ of all bottlenecks with relation $S \prec S'$ if and only if $I(S,\cP) \subseteq I(S',\cP)$ forms a partial order (a assuming that any ties, i.e., when $I(S,\cP)=I(S',\cP)$ are resolved arbitrarily).
Let $S_1,\ldots,S_t$ be the maximal elements with respect to this partial order.
Note that for $i\neq j$, $I(S_i,\cP)\cap I(S_j,\cP) = \emptyset$ and for any bottleneck $S'\notin\{S_1,\ldots,S_t\}$ we have $I(S',\cP)\subseteq I(S_i,\cP)$ for some $i\in\{1,\ldots,t\}$.
Assume without loss of generality that the `maximal' bottlenecks are ordered according to the left endpoints of their intervals,
\[\botStartS{S_1} \leq \botStartS{S_2} \leq \cdots \leq \botStartS{S_t}.\]

We show how to arrive at the desired state $s$.
To that end we argue, by induction on $j$, that for each
\[j\in J:=\{1,\ldots,l\}\setminus\bigcup_{i=1}^{k}\left\{\botStartS{S_i},\ldots,\botEndS{S_i}-1\right\}\]
there exists a state $s_j$ such that $\cover{s_j}=G[X_1\cup\cdots\cup X_j]$, $\bag{s_j}=X_j$ and $Tab[s_j]=true$.

Since the first bag of $\cP$ consists of a vertex in $\cI$, this clearly holds for $j=1$ so take $j>1$ and assume that the claim is true for each $j'\in J\cap\{1,\ldots,j-1\}$.
We consider two cases.
In the first case suppose that $j\notin I(S_i,\cP)$ for each $i\in\{1,\ldots,t\}$.
Hence we have $j-1\in J$.
This implies, according to Lemma~\ref{lem:small}, that for each bottleneck $S$, either there are at most $2k$ $S$-branches $H$ such that $\graphStartS{H}\leq j$, or there are at most $2k$ $S$-branches such that $j\leq\graphEndS{H}$.
Thus, there exists a state $s_{j}$ such that $\cover{s_{j}}=G[X_1\cup\cdots\cup X_{j}]$ and $\bag{s_{j}}=X_{j}$.
The step extension rule and $Tab[s_{j-1}]=true$, which holds by the induction hypothesis, imply $Tab[s_{j}]=true$ as required.

In the second case we have $j\in I(S_i,\cP)$ for some bottleneck $S_i$.
By the definition of the set $J$, $j=\botEndS{S_i}$.
Hence, the preceding index of $j$ in the set $J$ is $j'=\botStartS{S_i}-1$.
Again, by the definition of $I(S_i,\cP)$, Lemma~\ref{lem:small}, and the maximality of $S_i$ with respect to the partial order, we have that for each bottleneck set $S$ either at most $2k$ $S$-branches $H$ satisfy $\graphStartS{H}\leq j'$, or at least $|\branches{S_{i'}}|-2k$ $S$-branches $H$ are contained in $[1,j']$, i.e., satisfy $\graphEndS{H}\leq j'$, depending whether $\botEndS{S}\leq j'$ or $\botStartS{S}>j'$.
Thus, there exists a state $s_{j}$ such that $\cover{s_j}=G[X_1\cup\cdots\cup X_j]$ and $\bag{s_j}=X_j$.
Consider the jump extension rule constructed for $S'=S_i$.
For the set $B_{S'}$ in~\ref{jump:green} and~\ref{jump:otherS} take all $S'$-branches that are not in-branches, i.e., those that are covered in $[j'+1,j]$ in $\cP$.
Note that each $S$-branch of each bottleneck $S\neq S_i$ such that $S\subseteq X_{j'}$ waits in the interval $I(S_i,\cP)$ because $\cP$ is $S_i$-structured, which ensures the condition \ref{jump:other}.
Condition~\ref{jump:rec} holds because the decomposition $\cP_H$ in~\ref{jump:rec} exists which is certified by the decomposition $\cP$, namely $\cP_H=(X_{\graphStartS{H}}\cap V(H),\ldots,X_{\graphEndS{H}}\cap V(H))$.
Thus, $Tab[j']=true$ (which holds by the induction hypothesis) ensures that $Tab[j]=true$.

Finally observe that $l\in J$, $\cover{s_{l}}=G[X_1\cup\cdots\cup X_l]=V(G)$ and $Tab[s_l]=true$.
Thus, $s=s_l$ is the required state.
\end{proof}

\begin{lemma} \label{lem:induction-tab-todec}
For any state $s$, if $Tab[s]=true$, then $G_s$ has an $\cI$-connected path decomposition of width at most $k-1$.
\end{lemma}

\begin{proof}
Proof by induction on the position of $s$ in the ordering $\precs$. First, let $\cover{s} = \{v\}$ for some $v \in V(G)$ (notice that such states are smallest, according to $\precs$).
If $v \in \cI$, then $Tab[s]$ was set true in the initialization step. This is justified by considering connected path decomposition consisting of a single bag $\{v\}$, which is a proper connected path decomposition of the single-vertex graph $(\{v\},\emptyset)$.
On the other hand, if $v \notin \cI$, then $Tab[s]$ is never set $true$, as the extension rules apply only to states with $|\cover{s}| > 1$.
This is also correct, as no decomposition of $G_s$ contains a vertex of $\cI$.

Now suppose that $|\cover{s}| \geq 2$, and the Lemma holds for all states $w \precs s$. Since $Tab[s]=true$, its value must have been set by one of the extension rules. Consider two cases.

\paragraph{Case 1: $Tab[s]$ was set by step extension.}
Consider the state $w$. If $\bag{s} \not\subseteq \bag{w}$, then since by~\ref{step:4}, $\bag{s} \cap \cover{w} \subseteq \bag{w}$, we have that $\bag{s} \not\subseteq \cover{w}$. Therefore, $\cover{s} = \cover{w} \cup \bag{s}$ implies that $|\cover{s}| > |\cover{w}|$, which means that $w \precs s$.
On the other hand, if $\bag{s} \subseteq \bag{w}$, we have $\cover{s} = \cover{w}$ due to~\ref{step:3}. However, since $s$ and $w$, are distinguishable, we have $\bag{s} \subsetneq \bag{w}$, so $|\bag{s}| < |\bag{w}|$ and thus $w \precs s$.

So, by the inductive assumption, $Tab[w]$ was set properly and there exists an $\cI$-connected path decomposition $\cP = (X_1,X_2,\ldots,X_l)$ of $G_w$ that has width at most $k-1$, where $X_l = \bag{w}$. Let $X_{l+1} := \bag{s}$ and let $\cP' := (X_1,X_2,\ldots,X_l,X_{l+1})$.

We claim that $\cP'$ is an $\cI$-connected path decomposition of $G_s$.
Indeed, $\bigcup_{i=1}^{l+1} X_i = \bigcup_{i=1}^{l} X_i \cup X_{l+1} = \cover{w} \cup \bag{s} = \cover{s}$. 
Now, consider an edge $vu$ of $G_s$. If both $v,u$ belong to $\cover{w}$, they appear in some $X_i$ for $i \leq l$ (by the inductive assumption). If both $v,u$ belong to $\bag{s}$, we are done too. Finally, if $v \in \cover{w}$ and $u \in \bag{s} \setminus \cover{w}$, we know from~\ref{step:2} that $v \in \border{w} \subseteq \bag{s}$, so we are again in the previous case.
Now suppose for a contradiction that there are some $1 \leq i < j \leq l$, such that $X_i \cap X_{l+1} \not\subseteq X_j$. This means that $\bag{s}=X_{l+1}$ contains a vertex of $\cover{w} \setminus \bag{w}$, which is a contradiction with~\ref{step:4}.

Moreover, since $\width{\cP} \leq k-1$ and $|\bag{s}| \leq k$, we have $\width{\cP'} \leq k-1$. Finally, since $\cP$ is $\cI$-connected and (according to~\ref{step:1}) every connected component of $X_{l+1}$ either contains a vertex of $X_l$, or is adjacent to one, or belongs to $\cI$, we observe that $\cP'$ is also $\cI$-connected.
Finally, we immediately obtain from induction hypothesis that $\cP'$ starts with $\cI$.
This justifies setting $Tab[s]=true$.

\paragraph{Case 2: $Tab[s]$ was set by jump extension.}
Let $w=(X,\{B_S\}_S,\{f^B_S\}_S)$, $s=(X,\{B_S\}_S,\{g^B_S\}_S)$ and let $S' \subseteq X$ be defined as in the definition of the jump extension. To simplify the notation, set $\cB' := \branches{S'} \setminus B_{S'} = \{H_1,H_2,\ldots,H_m\}$.
Observe that since $S'$ is a bottleneck, we have $|\branches{S'}| \geq 2k+1$, thus there is at least one $H \in \cB'$.
Since $V(H) \not\subseteq \cover{w}$ and $V(H) \subseteq \cover{s}$ by~\ref{jump:green} and \ref{jump:other}, we have $|\cover{w}| < |\cover{s}|$ and thus $w \precs s$. So, by the inductive assumption, $Tab[w]$ was set properly to be $true$ and there exists an $\cI$-connected path decomposition $\cP = (X_1,X_2,\ldots,X_l)$ of $G_w$ with width at most $k-1$, where $X_l = X = \bag{w}$. By~\ref{jump:rec}, for every $H \in \cB'$ there is a path decomposition $\cP_H = (X^H_1,X^H_2,\ldots,X^H_{l(H)})$ of width at most $k-|X|-1$, such that $X^H_1$ contains a neighbor of $S'$, i.e., $\cP_H$ is $N_G(S')\cap V(H)$-connected.

We claim that
\[ \cP' = \cP \circ \left(\prod_{i=1}^m \prod_{j=1}^{l(H_i)} (X_l \cup X^{H_i}_j)\right) \circ X_l,
\]
where both $\circ$ and $\prod$ denote concatenation of appropriate sequences, is an $\cI$-connected path decomposition of $G_s$ of width at most $k-1$.

First, observe that $\cover{s} = \cover{w} \cup \bigcup_{H \in \cB'} V(H)$ due to~\ref{jump:otherS} and~\ref{jump:other}. By the definition of $\cP$ and decompositions $\cP_H$ for $H \in \cB'$, we observe that $\cP'$ covers exactly $\cover{s}$.

Now consider an edge $vu$ of $G_s$. If both vertices $v,u$ belong to $\cover{w}$, or to $V(H)$ for some $H \in \cB'$, then, by the definition of $\cP$ and $\cP_H$, both $v$ and $u$ appear in some bag of the decomposition $\cP'$.
If $v \in \cover{w}$ and $u \in V(H)$ for some $H \in \cB'$, then we know that $v \in \border{w}$ and therefore $v \in X_l$, so both vertices appear in every bag containing $u$. Finally, we observe that there are no edges joining vertices from different $S'$-branches.

The third condition of the definition of path decomposition follows directly from the definition of $\cP$ and $\cP_H$ and the fact that subgraphs $H$ are $S'$-branches.

Observe that $|X_i| \leq k-1$ for $i \leq l$ (by the definition of $\cP$), and since $|X^H_j| \leq k-|X|-1$ for every $H$ and $j$, we have $| X \cup X^H_j| \leq k-1$, so $\width{\cP'} \leq k-1$.

Note that $\cP$ is $\cI$-connected according to the induction hypothesis. Moreover, each $\cP_H$ is $N_G(S')\cap V(H)$-connected for each $H \in \cB'$ for some $S' \subseteq X$.
Thus, $\cP'$ is $\cI$-connected.

This completes the proof.
\end{proof}

Combining Lemmas \ref{lem:induction-dec-totab} and \ref{lem:induction-tab-todec}, we obtain the following corollary.
\begin{corollary}\label{cor:correct}
The algorithm is correct, i.e., the value of $Tab[s]$ is $true$ for some state $s$ with $\cover{s}=V(G)$ if and only if $\cpw{G}\leq k-1$.
\qed
\end{corollary}

Now let us estimate the computational complexity of our algorithm.
\begin{lemma}\label{lem:complexity}
For every fixed $k \geq 1$, a graph $G$ with $n$ vertices, and $\cI \subseteq V$, there is an algorithm deciding in time $f(k) \cdot n^{O(k^2)}$ whether $G$ has an $\cI$-connected path decomposition of width at most $k-1$, where $f$ is a function depending on $k$ only.
\end{lemma}

\begin{proof}
We do induction on $k$.
First, observe that for a connected graph $G$, $\cpw{G} = 0$ if and only if $G$ is a single-vertex graph.
Moreover, $\cpw{G} = 1$ if and only if $G$ is a caterpillar, and optimal connected path decompositions of caterpillars have very simple structure, so we can verify in polynomial time whether there is an $\cI$-connected one.

So assume that $k \geq 2$ and that the claim holds for $k-1$.
For every vertex $s^* \in V$, we run the dynamic programming algorithm that we described in Section~\ref{sec:algorithm}.
The correctness of the algorithm follows from Corollary \ref{cor:correct}. Now let us estimate its computational complexity.
Recall that the total number of states is $O(n^{3k})$, so the total number of pairs of states is $O(n^{6k})$. For each pair of states we check if one of the two extension rules can be applied.

Observe that for each state $s$, we can compute $\cover{s}$, $\bag{s}$ and $\border{s}$ in polynomial time. Thus checking if the step extension can be applied can also be done in polynomial time.

Now consider the possible jump extension from a state $w$ to a state $s$.
Verifying the first three conditions can be clearly done in polynomial time. We check in~\ref{jump:rec} if the appropriate path decomposition $\cP_H$ of each $S'$-branch $H$ exists by calling the algorithm recursively with the initial set $N_G(S')\cap V(H)$. By the inductive assumption, this can be done in total time bounded by $n^{O(1)} \cdot f'(k-1) n^{c \cdot (k-1)^2}$, for some function $f'$ and a constant $c'$.
This gives total time complexity 
\[n^{O(1)} \cdot n^{6k} \cdot f'(k-1) \cdot n^{c' (k-1)^2} = f(k) \cdot n^{O(k^2)}\]
for some function $f$.
\end{proof}

Now, the main result of the paper follows easily from Lemma \ref{lem:complexity}.
\begin{theorem}\label{thm:main}
For every fixed $k \geq 1$, there is an algorithm deciding in time $f(k) \cdot n^{O(k^2)}$ whether $\cpw{G} \leq k-1$, for some function $f$ depending on $k$ only, i.e., in time polynomial in $n$.
\end{theorem}

\begin{proof}
For every vertex $s^* \in V$, we run the dynamic programming algorithm for $\cI = \{s^*\}$, i.e., we exhaustively guess a vertex in the first bag of some fixed solution.
By Lemma \ref{lem:complexity}, the total running time is as claimed.
\end{proof}

Let us point out that we did not try to optimize the dependence of the degree of the polynomial function in Theorem \ref{thm:main} on $k$, as we were only interested in finding a polynomial algorithm.

\section{Open problems} \label{sec:open-problems}

As pointed out, both pathwidth and connected pathwidth are asymptotically the same for an arbitrary graph $G$, namely $\cpw{G}/\pathwidth{G}\leq 2+o(1)$.
However, there are several open questions regarding the complexity of exact algorithms for connected pathwidth.
One such immediate question that is a natural next step in the context of our work is whether connected pathwidth is FPT with respect to this parameter?
Also, it is not known if connected pathwidth can be computed faster than in time $O^*(2^n)$ for an arbitrary $n$-vertex graph (recall that this is possible for pathwidth).

The notion of connected pathwidth appeared in the context of pursuit-evasion games called node search, edge search or mixed search.
A challenging and long-standing open question related to those games is whether their connected variants belong to NP?
See~\cite{BFFFNST12} for more details regarding this question.

\section*{Acknowledgements}
This research has been partially supported by National Science Centre (Poland) grant number 2015/17/B/ST6/01887.

\bibliographystyle{abbrv}
\bibliography{sample}

\end{document}